\documentclass[10pt,final,twocolumn]{IEEEtran}
\usepackage{cite}
\usepackage{psfrag}
\usepackage{latexsym, amsmath, color, amsfonts, amssymb,graphicx}
\usepackage{algorithm}
\usepackage{algorithmic}

\newtheorem{theorem}{Theorem}[section]
\newtheorem{lemma}[theorem]{Lemma}

\newtheorem{proposition}[theorem]{Proposition}
\newtheorem{remark}[theorem]{Remark}

\newcommand{\R}{{\rm  I\kern-2pt R}}

\begin{document}
	\title{LQG Control Over SWIPT-enabled \\Wireless Communication Network}
	\author{
		Huiwen~Yang,~Lingying~Huang,~Yuzhe~Li,~Subhrakanti~Dey,~and~Ling~Shi %\IEEEmembership{Member, IEEE}
		%\thanks{\textcolor{black}{This work was supported in part by \ldots.} }
		%	\thanks{\textcolor{black}{H. Yang, M. Huang, and L. Shi are with the Department of Electronic and Computer Engineering, Hong Kong University of Science and Technology, Hong Kong (email: hyangbr@connect.ust.hk, mhuangak@connect.ust.hk, eesling@ust.hk).} }
		%Huiwen Yang, Mengyu Huang, and Ling Shi are with the Department of Electronic and Computer Engineering, Hong Kong University of Science and Technology, Hong Kong (email: hyangbr@connect.ust.hk, mhuangak@connect.ust.hk, eesling@ust.hk).
		%	\thanks{\textcolor{black}{Y. Li is with the State Key Laboratory of Synthetical Automation for Process Industries, Northeastern University, Shenyang 110004, China (e-mail: yuzheli@mail.neu.edu.cn).}}
		%	\thanks{\textcolor{black}{S. Dey is with the Hamilton Institute, National University of Ireland, Maynooth, Ireland (email: Subhra.Dey@mu.ie).}}
	}
	
	%{YANG \MakeLowercase{\textit{et al.}}: Joint Power Allocation for Remote State Estimation with SWIPT}
	
	\maketitle

	\begin{abstract}
		\textcolor{black}{In this paper, we consider using simultaneous wireless information and power transfer (SWIPT) to recharge the sensor in the LQG control, which provides a new approach to prolonging the network lifetime. 
			We analyze the stability of the proposed system model and show that there exist two critical values for the power splitting ratio $\alpha$. 
			%We show that there exist two critical values for the power splitting ratio $\alpha$, and the cost of the infinite horizon linear-quadratic-Gaussian (LQG) control is bounded if and only if $\alpha$ is between these two critical values. 
			Then, we propose an optimization problem to derive the optimal value of $\alpha$. This problem is non-convex but its numerical solution can be derived by our proposed algorithm efficiently. Moreover, we provide the feasible condition of the proposed optimization problem.
			Finally, simulation results are presented to verify and illustrate the main theoretical results.}
	\end{abstract}
	%
	%% Note that keywords are not normally used for peerreview papers.
	\begin{IEEEkeywords}
		\textcolor{black}{Networked control systems, SWIPT, packet drop, stability.}
	\end{IEEEkeywords}
	
	\IEEEpeerreviewmaketitle
	
\section{Introduction}
	Sensing and actuation are essential factors for the control of complex dynamical networks. 
	%As the most fundamental optimal estimator and controller, Kalman filter and linear-quadratic regulator (LQR) have been studied for decades. 
	As one of the most fundamental optimal controller in control theory, the linear-quadratic-Gaussian (LQG) controller, which is a combination of a Kalman filter with a linear-quadratic regulator (LQR), has been studied for decades~\cite{nilsson1998stochastic, anderson2012optimal,  schenato2007foundations, gupta2007optimal, xu2022channel}. With the extensive use of wireless devices (e.g., sensors, actuators, and remote controllers), sensing signals and control signals are transmitted over wireless communication networks, where packets may be lost or delayed. The effect of packet loss on Kalman filtering was studied in the seminal paper~\cite{sinopoli2004kalman}. Later, the impact of the network reliability on control and estimation was comprehensively and systematically analyzed by~\cite{schenato2007foundations}. 
	%Subsequently, many works further investigated the estimation and control problems under different communication network settings~\cite{bibid}. The performances of Kalman filtering over packet-delaying networks and packet-dropping networks were analyzed from a probabilistic perspective in~\cite{shi2009kalman} and~\cite{shi2010kalman}, respectively. 

	In real applications, sensors are usually battery-powered and their battery capabilities are limited. As a result, sensor scheduling and power management are important for reducing power consumption and prolonging the lifetime of a network. 
	Many existing works have considered the sensor scheduling and power control problems for remote state estimation with Kalman filter~\cite{shi2011sensor, han2013online, ren2013optimal, wu2012event, han2015stochastic, li2013optimal, li2014multi, wu2015data, ding2017multi, ding2020dynamic, huang2021joint}. 
	The sensor scheduling with limited communication energy was investigated by~\cite{shi2011sensor, han2013online, ren2013optimal}. A deterministic event-based scheduling mechanism to achieve the trade-off between communication rate and estimation quality was proposed in~\cite{wu2012event}, and it was extended to a stochastic event-triggering mechanism in~\cite{han2015stochastic}. The power control problems under more practical communication models were studied by~\cite{li2013optimal, li2014multi, wu2015data, ding2017multi, ding2020dynamic}. To further increase the life span of networks, many researchers considered employing energy-harvesting sensors~\cite{nayyar2013optimal, nourian2014optimal, li2014transmission}. As a result of the development of radio frequency (RF) energy harvesting circuit design, wireless sensors equipped with an RF energy receiver have the ability to harvest energy from the RF signals transmitted by some wired-supplied devices, which may have unlimited energy supply but have limited transmission power. Different from the energy harvested from the external environments (e.g., solar energy, wind energy, body heat, etc.), the energy harvested from RF signal is more predictable and controllable, since the transmission power of the energy supplier can be controlled and adjusted. 
	\textcolor{black}{In the era of the fifth generation (5G) of wireless communication, there is an increasing demand for a technology that can transfer both information and power simultaneously to the end-devices~\cite{8214104}. As a result, the concept of simultaneous wireless information and power transfer (SWIPT) was introduced in~\cite{4595260}. In recent years, SWIPT has aroused great interest in wireless communication networks~\cite{8306833,8444683,8951078,9070210,9080059,9398230}, and has also been considered in networked control systems~\cite{yang2022joint}.} 
	In~\cite{yang2022joint}, a novel power control problem for remote state estimation was formulated, where the signal transmitted by the remote estimator served as the energy source of the sensor, and the transmission power allocations of the remote estimator and the sensor are jointly optimized. In this paper, we consider applying SWIPT to the LQG control. The control signals are transmitted by a transmitter, which has maximum transmission power. One part of the transmitted power will be used to decode the expected control input, and the other part of the transmitted power will be harvested by the sensor for the subsequent transmission of its measurements.
	
	The contributions of this paper are multi-fold. First, this is the first work that considers using SWIPT to recharge the sensor in the LQG control, which provides a new approach to prolonging the network lifetime. Second, we show that there exist two critical values for the power splitting ratio $\alpha$, and the cost of the infinite horizon LQG control is bounded if and only if $\alpha$ is between these two critical values. Third, we propose an optimization problem to derive the optimal value of $\alpha$. This problem is non-convex but its numerical solution can be derived by our proposed algorithm efficiently. 
	
	\textcolor{black}{
		The remainder of this paper is organized as follows.
		Section \uppercase\expandafter{\romannumeral2} presents the system model and the communication model. Section \uppercase\expandafter{\romannumeral3} introduces necessary preliminaries and provides the main theoretical results. Section \uppercase\expandafter{\romannumeral4} presents the simulation results to verify the main theorems in Section \uppercase\expandafter{\romannumeral3}. Section \uppercase\expandafter{\romannumeral5} concludes this paper and presents some future work.}
	
	\textcolor{black}{
		\emph{Notations:} 
		$\mathbb{R}$ is the set of real numbers, $\mathbb{R}^n$ is the $n$-dimensional Euclidean space, and $\mathbb{R}^{n\times m}$ is the set of real matrices with size $n\times m$. 
		For a matrix $X$, $X > 0$ $(X\geq 0)$ denotes $X$ is a positive definite (positive semidefinite) matrix, $\lambda_i^u(X)$ denotes the unstable eigenvalues of $X$, and $\text{Tr}(X)$ denotes the trace of $X$. 
		%$\Vert \cdot \Vert \triangleq \sigma_\text{max}(\cdot)$ is the spectral norm of a matrix, where $\sigma_\text{max}(\cdot)$ represents the largest singular value of the matrix. 
		%$\mathcal{N}(m,\Sigma)$ represents the Gaussian distribution with $m$-mean vector and covariance matrix $\Sigma$. 
		$\mathbb{E}[\cdot]$ is the expectation of a random variable. $\mathbb{P}(\cdot | \cdot)$ refers to conditional probability.
	}

	\section{Problem Setup}
	\subsection{System Model}
	\begin{figure}[h]
		\centering
		\includegraphics[width=3.4in]{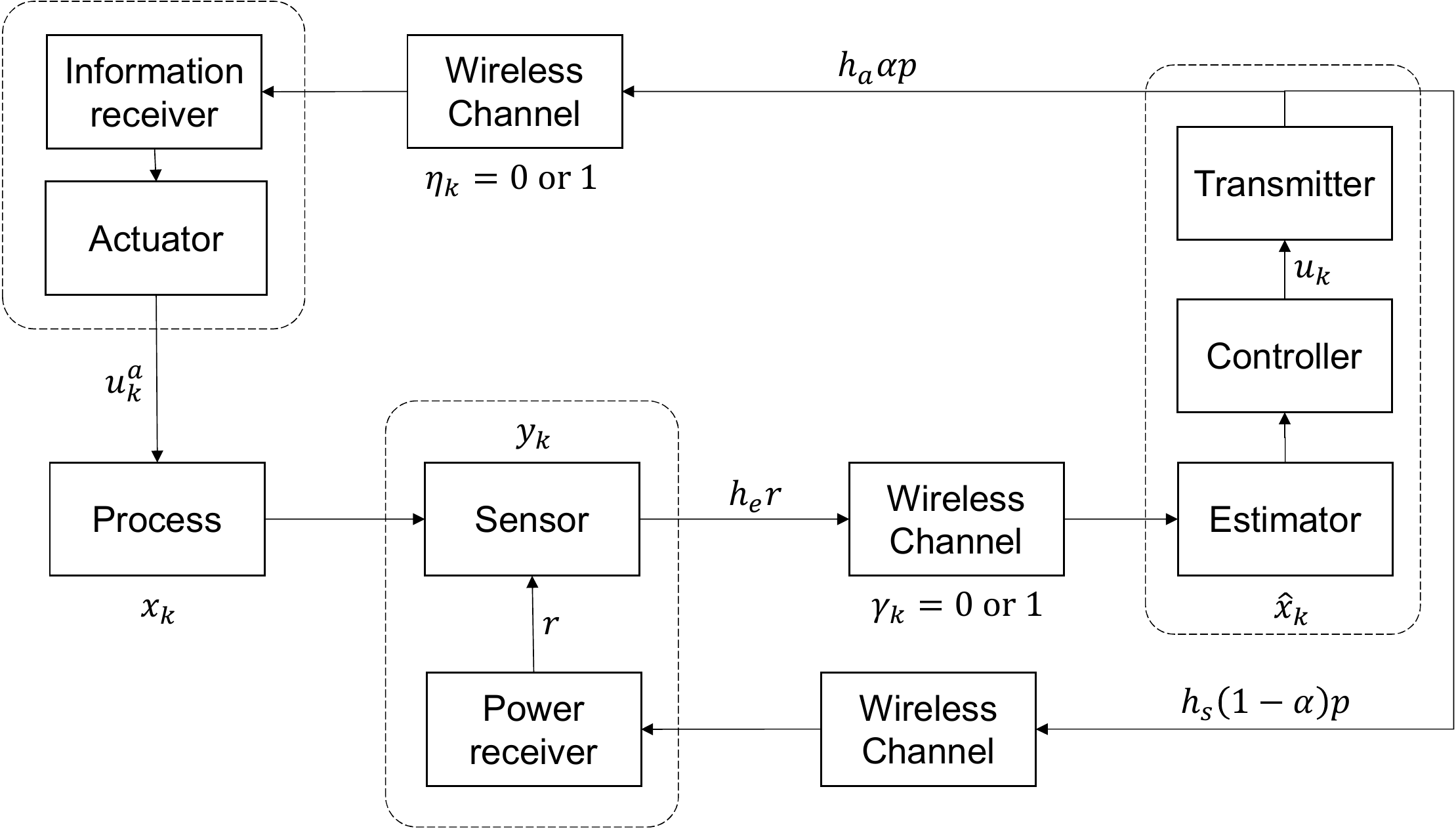}
		\caption{System Model}
		\label{system}
	\end{figure}
	Consider the following discrete-time LTI system:
	\begin{equation}
		x_{k+1} = A x_k + B u^a_k + w_k,
	\end{equation}
	where $A\in \mathbb{R}^{n\times n}$ and $B\in\mathbb{R}^{n\times q}$ are constant matrices, $x_k \in \mathbb{R}^n$ is the system state, $u^a_k\in\mathbb{R}^q$ is the actual control input exerted to the system by the actuator, 
	$w_k \in \mathbb{R}^n$ is a zero mean Gaussian noise with covariance $Q\geq 0$. Moreover, we assume $x_0$ is Gaussian with mean $\bar{x}_0$ and covariance $P_0$.
	
	The controller and the estimator are colocated with a transmitter, which transmits control signals to the actuator and transfers power to the sensor simultaneously.
	The control signal $u_k\in\mathbb{R}^q$ is sent over a lossy channel. A sequence of independent and identically distributed (i.i.d.) Bernoulli random variables, i.e., $\{\eta_k\}$, is used to model the control packet loss. If the control signal is successfully received by the actuator at time $k$, then $\eta_k=1$; otherwise, $\eta_k=0$. Then we have 
	\begin{equation}
		u^a_k = \eta_k u_k.
	\end{equation}
	
	The sensor can obtain the measurement of the system, i.e., 
	\begin{equation}
		y_k = C x_k + v_k,
	\end{equation}
	where $C\in\mathbb{R}^{p\times n}$ is a constant matrix, $y_k \in \mathbb{R}^p$ is the sensor measurement, and $v_k\in \mathbb{R}^p$ is a zero mean Gaussian noise with covariance $R > 0$. 
	%It runs a Kalman filter locally to obtain a local state estimate $\hat{x}^s_k = \mathbb{E}[x_k|Y_k]$, where $Y_k = \{y_1, \ldots, y_k\}$. 
	It sends $y_k$ to the estimator over a lossy channel, which can be modeled by a sequence of i.i.d. Bernoulli random variables $\{\gamma_k\}$. If $y_k$ is received successfully by the estimator, then $\gamma_k = 1$; otherwise, $\gamma_k = 0$.
	
	%The estimator and controller unit:
	In this paper, we consider that the transmission control protocol (TCP) is adopted by the system, which means the senders can know whether the packet delivery is successful or unsuccessful within the same sampling time period. Denote the information set available for the estimator at time $k$ by $\mathcal{L}_k = \{\gamma_1 y_1, \ldots, \gamma_k y_k, \gamma_1, \ldots, \gamma_k, \eta_1, \ldots, \eta_{k-1}\}$. The estimator needs to determine the optimal state estimation $\hat{x}_k$ based on $\mathcal{L}_k$, and the controller needs to calculate the optimal control input $u_k$ based on $\hat{x}_k$. 
	
	\subsection{Simultaneous Wireless Information and\\ Power Transfer}
	By simultaneous wireless information and power transfer (SWIPT), the transmitter at the estimator/controller side can transmit the control signal to the actuator and transmit power to the sensor at the same time. The transmitter uses a constant transmission power $p$.
	%constraint
	%\begin{equation}
	%	p_k \leq p_\text{max},
	%\end{equation}
	%where $p_k$ is the transmission power at time $k$ and $p_\text{max}$ is the maximum transmission power.
	The portion of the power used for transmitting the control signal is denoted by $\alpha p$, and the portion for recharging the sensor is $(1-\alpha)p$. The actuator will decode the received control signal and the sensor will use all the harvested power to transmit its local estimate at each time.
	
	The probability of successful control signal reception for the actuator is given by
	\begin{equation}\label{}
		\mathbb{P}(\eta_{k}=1|h_a, \alpha, p)\triangleq \eta(h_a\alpha p),
	\end{equation}
	where $h_a$ is the channel fading of the channel from the transmitter to the actuator, and $\eta(\cdot): \left[ \right. 0, \infty \left)\right.\rightarrow [0,1]$ is a monotonically increasing continuous function, which is decided by the particular digital modulation mode and the channel state from the controller to the actuator. For simplicity, we just use $\eta$ to denote this probability when the power $p$ and the ratio $\alpha$ are fixed.
	
	The harvested energy $r$ can be characterized by a
	function of the received power $h_s(1-\alpha)p$, i.e.,
	\begin{equation}
		r = \psi(h_s((1-\alpha)p)),
	\end{equation}
	where $h_s$ is the channel fading of the channel from the transmitter to the sensor, and $\psi$ is a nondecreasing function determined by the energy harvesting circuit.
	Similarly, the probability of successful reception of the sensor's local estimate for the estimator is given by
	\begin{equation}\label{}
		\mathbb{P}(\gamma_{k}=1|h_e,r)\triangleq \gamma(h_e r),
	\end{equation}
	where $h_e$ is the channel fading of the channel from the sensor to the estimator, and $\gamma(\cdot): \left[ \right. 0, \infty \left)\right.\rightarrow [0,1]$ is also a monotonically increasing continuous function decided by the particular digital modulation mode and the channel state from the sensor to the estimator, and we use $\gamma$ to denote this probability for simplicity. 
	
	\section{LQG Control with SWIPT}\label{sec:finite_lqg}
	%As we have assumed that the smart sensor runs a Kalman filter locally to obtain the optimal estimate $\hat{x}^s_k$ and the corresponding estimate error covariance as follows:
	\subsection{Prelimimaries of LQG control}
	\textcolor{black}{In this subsection, we present some useful lemmas, which are necessary for the derivation of the main results.}
	
	We first define the following variables:
	\textcolor{black}{
		\begin{align}
			%	\begin{split}
			\hat{x}_{k}&\triangleq\mathbb{E}[x_k|
			\mathcal{L}_k],\\
			P_{k}&\triangleq \mathbb{E}[(x_k-\hat{x}_{k})(x_k-\hat{x}_{k})^T|\mathcal{L}_k].
			%	\end{split}
		\end{align}
	}
	The optimal estimator can be derived as
	\textcolor{black}{
		\begin{align}
			%	\begin{split}
			\hat{x}_{k+1}^{-}&\triangleq A\mathbb{E}[x_k|\mathcal{L}_k]+\eta_{k}Bu_k = A\hat{x}_{k}+\eta_{k}Bu_k,\\
			P_{k+1}^{-}&\triangleq \mathbb{E}[(x_{k+1}-\hat{x}_{k+1}^{-})(x_{k+1}-\hat{x}_{k+1}^{-})^T|\mathcal{L}_k]\\ & =AP_{k}A^T+Q.
			%	\end{split}
		\end{align}
	}
	
	% 	\begin{lemma}\cite[Section \uppercase\expandafter{\romannumeral4}.A]{schenato2007foundations}
	For the TCP-like system, the optimal estimator is the following:
	\begin{align}
		\hat{x}_{k+1} &= \hat{x}_{k+1}^{-} + \gamma_{k+1}K_{k+1}(y_{k+1}-C\hat{x}_{k+1}^{-}),\\
	\begin{split}
			    P_{k+1}^{-}&\triangleq  \mathbb{E}[(x_{k+1}-\hat{x}_{k+1}^{-})(x_{k+1}-\hat{x}_{k+1}^{-})^T|\mathcal{L}_k]\\ & =AP_{k}A^T+Q.
			\end{split}	
        \end{align}
	Consider the following cost function:
	\begin{equation}
		\begin{split}
			J_N(\bar{x}_0, P_0) =& \mathbb{E}[x_N^T W_N x_N \\&+ \sum_{k=0}^{N-1}(x_k^T W_k x_k + {u^a_k}^T U_k u^a_k)| \bar{x}_0, P_0],
		\end{split}
	\end{equation}
	then, to obtain the optimal control input sequence, the following optimization problem should be solved
	\begin{equation}\label{flqg}
		\min_{u_k, k=0,1,\ldots,N-1} J_N(\bar{x}_0, P_0).
	\end{equation}
	Define the optimal value function $V_k(x_k)$ as follows:
	\begin{align}
		V_N(x_N)&\triangleq \mathbb{E}[x_N^TW_N x_N|\mathcal{L}_N],\\
		V_k(x_k)&\triangleq \min_{u_k}\mathbb{E}[x_k^TW_k x_k+\eta_ku_k^TU_ku_k+V_{k+1}(x_{k+1})|\mathcal{L}_k],\label{value_function}
	\end{align}
	where $k=1,2,\ldots, N-1$. Using dynamic programming, it can be shown that $J_N^*\triangleq \min_{u_k, k=0,1,\ldots,N-1} J_N(\bar{x}_0, P_0) = V_0(x_0)$. For the TCP-like system, the following lemma holds:
	\begin{lemma}\cite[Lemma 5.1]{schenato2007foundations}
		The value function $V_k(x_k)$ in (\ref{value_function}) has the following form:
		\begin{equation}
			V_k(x_k)=\mathbb{E}[x_k^TS_k x_k|\mathcal{L}_k]+c_k, k=0,1,\ldots,N,
		\end{equation}
		where
		\begin{align}
			\begin{split}
				S_k=&A^TS_{k+1}A+W_k\\&-\eta A^TS_{k+1}B(B^TS_{k+1}B+U_k)^{-1}B^TS_{k+1}A,
			\end{split}	\\
			\begin{split}
				c_k =& \mathrm{Tr}\left((A^TS_{k+1}A+W_k-S_k)P_{k}\right)+\mathrm{Tr}\left(S_{k+1} Q\right)\\&+\mathbb{E}[c_{k+1}|\mathcal{L}_{k}],
			\end{split}
		\end{align}
		with initial values $S_N=W_N$ and $c_N=0$, and the optimal control input is given by
		\begin{equation}
			u_k=-(B^TS_{k+1}B+U_k)^{-1}B^TS_{k+1}A\hat{x}_{k}=L_k \hat{x}_{k}.
		\end{equation}
	\end{lemma}
	Since $J_N^*(\bar{x}_0,P_0)=V_0(x_0)$, we have 
	\begin{equation}\label{cost}
		\begin{split}
			J_N^*(\bar{x}_0,P_0) = &\bar{x}_0^TS_0 \bar{x}_0 + \mathrm{Tr}\left(S_0 P_0\right)+\sum_{k=0}^{N-1}\mathrm{Tr}\left(S_{k+1}Q\right)\\&+\sum_{k=0}^{N-1}\mathrm{Tr}\left((A^TS_{k+1}A+W_k-S_k)\mathbb{E}_\gamma[P_{k}]\right),
		\end{split}
	\end{equation}
	where $\mathbb{E}_\gamma[\cdot]$ explicitly indicates that the expectation is calculated with respect to the arrival sequence $\{\gamma_k\}$.
	
	%\subsection{Infinite Horizon LQG Control with SWIPT}
	The infinite horizon LQG can be obtained by taking the limit for $N\rightarrow +\infty$ of the previous equations in Section \ref{sec:finite_lqg}. 
	
	Define the Modified Riccati Algebraic Equation (MARE) as 
	\begin{equation}\label{MARE}
		\begin{split}
			S=\Pi(S,A,B,W,U,\eta),
		\end{split}
	\end{equation}
	where $\Pi(S,A,B,W,U,\eta)\triangleq A^TSA+W -\eta A^TSB(B^TSB+U)^{-1}B^TSA$. Then, we have the following lemmas.
	
	\begin{lemma}\cite[Lemma 5.4]{schenato2007foundations}\label{lm1}
		Consider the modified Riccati equation defined in (\ref{MARE}). Let $A$ be unstable, $(A,B)$ be controllable, and $(A,W^{1/2})$ be observable. Then, the MARE has a unique strictly positive definite solution $S_\infty$ if and only if $\eta>\eta_c$, where $\eta_c$ is the critical arrival probability defined as
			\begin{equation}\label{def:eta}
				\eta_c\triangleq\inf_\eta\{0<\eta<1|S=\Pi(S,A,B,W,U,\eta), S\geq 0\}.
			\end{equation}
		% \begin{enumerate}
		% 	\item[a)] The MARE has a unique strictly positive definite solution $S_\infty$ if and only if $\eta>\eta_c$, where $\eta_c$ is the critical arrival probability defined as
		% 	\begin{equation}\label{def:eta}
		% 		\eta_c\triangleq\inf_\eta\{0<\eta<1|S=\Pi(S,A,B,W,U,\eta), S\geq 0\}.
		% 	\end{equation}
			% \item[b)] The critical probability can be numerically computed via the solution of the following quasiconvex LMIs optimization problem.
			% \begin{equation}
			% 	\eta_c = \arg\min_{\eta} \Psi_\eta(Y,Z)>0,\quad 0\leq Y\leq I,
			% \end{equation}
			% where
			% \begin{align*}
			% 	\begin{split}
			% 		&\Psi_\eta(Y,Z)=\\
			% 		&\begin{bmatrix}
			% 			Y & Y & \sqrt{\eta}ZU^{\frac{1}{2}} & \sqrt{\eta}T^T & \sqrt{\bar{\eta}}YA^T \\
			% 			Y & W^{-1} & 0 & 0 & 0 \\
			% 			\sqrt{\eta}U^{\frac{1}{2}}Z^T & 0 & I & 0 & 0 \\
			% 			\sqrt{\eta}T & 0 & 0 & Y & 0 \\
			% 			\sqrt{\bar{\eta}}AY & 0 & 0 & 0 & Y
			% 		\end{bmatrix},
			% 	\end{split}\\
			% 	&\bar{\eta}=1-\eta,\\
			% 	&T=AY+BZ^T.
			% \end{align*}
		% \end{enumerate}
	\end{lemma} 
	
	\begin{lemma}\cite[Theorem 5.5]{schenato2007foundations}\label{lm2}
		%	Consider the system (\ref{label}) and the
		Assume that $(A,Q^{1/2})$ is controllable, $(A,C)$ is observable, and $A$ is unstable. Then there exists a critical observation arrival probability $\gamma_c$, such that the expectation of estimator error covariance is bounded if and only if the observation arrival probability is greater than the critical arrival probability, i.e., 
		\begin{equation}\label{def:gamma}
			\mathbb{E}_\gamma[P_{k}]\leq M, \forall k \iff \gamma>\gamma_c,
		\end{equation}
		where $M$ is a positive definite matrix possibly dependent on $P_0$. Moreover, it is possible to compute a lower and an upper bound for the critical observation arrival probability $\gamma_c$, i.e., $p_\mathrm{min}\leq \gamma_c\leq\gamma_\mathrm{max}\leq p_\mathrm{max}$, where
		\begin{align*}
			p_\text{min}\triangleq&1-\frac{1}{\max_i|\lambda_i^u(A)|^2}, \quad
			p_\text{max}\triangleq 1-\frac{1}{\prod_i|\lambda_i^u(A)|^2},\\
			\gamma_\mathrm{max}\triangleq& \inf_\gamma\{0\leq \gamma\leq 1| P=\Pi(P,A^T,C^T,Q,R,\gamma), P\geq 0\}.
		\end{align*}
		%	The upper bound $\gamma_\mathrm{max}$ is given by the solution of the following optimization problem:
		%	\begin{equation}
		%		\gamma_\mathrm{max}=\arg\min_{\gamma}\Psi_\gamma(Y,Z)>0, \quad 0\leq Y \leq I,
		%	\end{equation}
		%	where 
		%	\begin{equation}
		%		\begin{split}
		%			&\Psi_\gamma(Y,Z)=
		%			\begin{bmatrix}
		%				Y & \sqrt{\gamma}(YA+ZC) & \sqrt{1-\gamma}YA \\
		%				\sqrt{\gamma}(A^TY+C^TZ^T) & Y & 0 \\
		%				\sqrt{1-\gamma}A^TY & 0 & Y
		%			\end{bmatrix}
		%		\end{split}
		%	\end{equation}
	\end{lemma}
	
	\subsection{Properties of optimal control}
	\textcolor{black}{As introduced in the last subsection, there exists a critical value for the probability of successful control signal reception, i.e., $\eta_c$ defined in~(\ref{def:eta}), and a critical value for the probability of successful reception of the sensor's local estimate, i.e., $\gamma_c$ defined in~(\ref{def:gamma}).} 
	In this subsection, we will show there exist two critical values, i.e., the left critical value $\underline{\alpha}$ and the right critical value $\overline{\alpha}$, for the power splitting ratio $\alpha$. \textcolor{black}{The existence of the two critical values for $\alpha$ results from that $\alpha$ realizes the trade-off between control and estimation.} 
	The infinite-horizon average cost of the LQG control will be unbounded if $\alpha$ is smaller than $\underline{\alpha}$ or larger than $\overline{\alpha}$. Moreover, the convergence property of the cost under the optimal control law is analyzed.

 Define the following functions: 
		%$h: \mathbb{S}^n_{+}\rightarrow \mathbb{S}^n_{+}$ and $\tilde{g}: [0,1]\times\mathbb{S}^n_{+}\rightarrow \mathbb{S}^n_{+}$ as
		\begin{align}
			\tilde{h}(X)\triangleq& AXA^T+Q,\\
			\hat{h}(X)\triangleq& A^TXA+W,\\
			\tilde{g}(\alpha, X)\triangleq& X^{-1}-\gamma(r) AXC^T(CXC^T+R)^{-1}CXA^T,\\
			\hat{g}(\alpha, X)\triangleq& X^{-1}-\eta(\alpha) A^TXB(B^TXB+U)^{-1}B^TXA.
			%	g_1(\gamma,X)\triangleq&\tilde{g}_1(\gamma,h(X)),\\
			%	g_2(\eta,X)\triangleq&h(\tilde{g}_2(\eta,X)).\\
		\end{align}
		For the sake of simplicity, we use $\tilde{g}_\alpha(X)$ and $\hat{g}_\alpha(X)$ to denote $\tilde{g}(\alpha, X)$ and $\hat{g}(\alpha,X)$, respectively. Also, define
		\begin{align}
			S(\alpha)=&\lim\limits_{k\rightarrow+\infty}(\hat{h}\circ\hat{g}_{\alpha})^k(S_N),\\
			\overline{P}(\alpha)=&\lim\limits_{k\rightarrow+\infty}(\tilde{g}_{\alpha}\circ\tilde{h})^k(P_0),
		\end{align}
  where $S_N=W$. 
  Then, we have the following lemma.
	\begin{lemma}
		If $\alpha_1\leq\alpha_2$, then $S(\alpha_1)\leq S(\alpha_2)$ and $\overline{P}(\alpha_1)\geq\overline{P}(\alpha_2)$.
	\end{lemma}
	\begin{proof}
            See Appendix~A.
	\end{proof}

	\begin{theorem}\label{thm0}
		The cost $J_N^*$ can be bounded as follows:
		\begin{equation*}
			J_N^\mathrm{min}\leq J_N^* \leq J_N^\mathrm{max},
		\end{equation*}
		where
		\begin{align}
			\begin{split}
				J_N^\mathrm{min} =& \bar{x}_0^T S_0 \bar{x}_0 + \mathrm{Tr}\left(S_0 P_0\right) + \sum_{k=0}^{N-1}\mathrm{Tr}\left(S_{k+1}Q\right)\\& +(1-\gamma(h_e\psi(h_s((1-\alpha)p))))  \\& \times \sum_{k=0}^{N-1}\mathrm{Tr}\left((A^TS_{k+1}A+W_k-S_k)\underline{P}_k\right),
			\end{split}
           \end{align}
           \begin{align}
			\begin{split}
				J_N^\mathrm{max} =& \bar{x}_0^T S_0 \bar{x}_0 + \mathrm{Tr}\left(S_0 P_0\right) + \sum_{k=0}^{N-1}\mathrm{Tr}\left(S_{k+1}Q\right) \\& +\mathrm{Tr}\left((A^TS_kA+W_k-S_k)\right.\\&\times\left(\overline{P}_k-\gamma(h_e\psi(h_s((1-\alpha)p)))\right.\\&\left.\left.\times \overline{P}_kC^T(C\overline{P}_kC^T+R)^{-1}C\overline{P}_k\right)\right),
			\end{split}				
		\end{align}
		and
		\begin{align}
			\underline{P}_k =& (1-\gamma(h_e\psi(h_s((1-\alpha)p))))A \underline{P}_{k-1}A^T+Q,\\
			\begin{split}
				\overline{P}_k =& A\overline{P}_kA^T+Q\\&-\gamma(h_e\psi(h_s((1-\alpha)p)) A\overline{P}_kC^T(C\overline{P}_kC^T+R)^{-1}\\&\times C\overline{P}_kA^T.
			\end{split}
		\end{align}
	\end{theorem}
	\begin{proof}
            See Appendix~B.
		% First, we have the following lemma:
		% \begin{lemma}%\cite[Lemma 5.4]{schenato2007foundations}
		% 	The expected error covariance matrix $\mathbb{E}_\gamma[P_{k}]$ satisfies the following bounds:
		% 	\begin{equation}
		% 		\tilde{P}_k \leq \mathbb{E}_\gamma[P_{k}] \leq 	\hat{P}_k,
		% 	\end{equation}
		% where
		% \begin{align}
		% 	\tilde{P}_k =& (1-\gamma(h_e\psi(h_s((1-\alpha)p))))\underline{P}_k, \\
		% 	\begin{split}
		% 	    \hat{P}_k =&  \overline{P}_k\\&-\gamma(h_e\psi(h_s((1-\alpha)p)))\overline{P}_kC^T(C\overline{P}_kC^T+R)^{-1}C\overline{P}_k.
		% 	\end{split}
		% \end{align}
		% \end{lemma}
		
		% \begin{proof}
		% 	This lemma can be proved based on the observation that the matrices $P_{k+1}^{-}$ and $P_k$ are concave and monotonic functions of $P_k^{-}$. The proof can be easily derived based on Lemma 5.2 in~\cite{schenato2007foundations} and is thus omitted.
		% \end{proof}
	 %    Then, Theorem~\ref{thm0} holds.
	\end{proof}

	\textcolor{black}{
		\begin{theorem}\label{thm1}
			Consider a system with $W_k = W, \forall k$ and $U_k = U, \forall k$. 
			There exist two critical values $\underline{\alpha}$ and $\overline{\alpha}$ such that the infinite-horizon average cost $J_\infty^*\triangleq\lim_{N\rightarrow \infty}\frac{1}{N}J_N^*(\bar{x}_0,P_0)$ can be bounded if and only if $\underline{\alpha} < \alpha < \overline{\alpha}$. Moreover, $J^\mathrm{min}_\infty\leq J^{*}_\infty\leq J^\mathrm{max}_\infty$, where
			%		$J_\infty^\mathrm{min}$ and $J_\infty^\mathrm{max}$ converge to the following values: 
			\begin{align}
				\begin{split}
					J_\infty^\mathrm{min}\triangleq&\mathrm{Tr}\left(SQ\right)+(1-\gamma(h_e\psi(h_s((1-\alpha)p))))\\&\times\mathrm{Tr}\left((A^TSA+W-S)\underline{P}\right),
				\end{split}\label{lowerbound}\\
				\begin{split}
					J_\infty^\mathrm{max}\triangleq&\mathrm{Tr}(SQ)+\mathrm{Tr}\left((A^TSA+W-S)\right.\\&\times\left(\overline{P}-\gamma(h_e\psi(h_s((1-\alpha)p)))\right.\\&\left.\left.\times \overline{P}C^T(C\overline{P}C^T+R)^{-1}C\overline{P}\right)\right),
				\end{split}\label{upperbound}
			\end{align}
			and 
			\begin{align}
				S =& A^TSA+W-\eta(h_a\alpha p) A^TSB(B^TSB+U)^{-1}B^TSA,\label{eq:s2}\\
				\underline{P} =& (1-\gamma(h_e\psi(h_s((1-\alpha)p))))A\underline{P}A^T+Q.\label{eq:up}\\
				\begin{split}
					\overline{P} =& A\overline{P}A^T+Q\\&-\gamma(h_e\psi(h_s((1-\alpha)p)) A\overline{P}C^T(C\overline{P}C^T+R)^{-1}C\overline{P}A^T.\label{eq:op}
				\end{split}
			\end{align}
			Moreover, the critical values $\underline{\alpha}$ and $\overline{\alpha}$ are the solutions of the following equations:
			\begin{align}
				\eta(h_a\underline{\alpha} p)=\eta_c,\label{eq1}\\
				\gamma(h_e\psi(h_s((1-\overline{\alpha})p)))=\gamma_c,\label{eq2}
			\end{align}
			where $\eta_c$ and $\gamma_c$ are the critical probabilities mentioned in Lemma \ref{lm1} and Lemma \ref{lm2}, respectively.
		\end{theorem}
		\begin{proof}
                See Appendix~C.
			% Following Lemma \ref{lm1} and Lemma \ref{lm2}, $\underline{\alpha}$ and $\overline{\alpha}$ must respectively satisfy equation (\ref{eq1}) and equation (\ref{eq2}). 
			% Since $\eta(\cdot)$ and $\gamma(\cdot)$ are monotonically increasing continuous functions, equation (\ref{eq1}) and equation (\ref{eq2}) have one unique solution, respectively. 
			% Then we have $\eta>\eta_c$ if and only if $\alpha>\underline{\alpha}$. 
			% % According to Lemma~\ref{lm1}, the MARE has a unique strictly positive definite solution $S_\infty$ if and only if $\alpha\leq\underline{\alpha}$. 
			% Similarly, we have $\gamma>\gamma_c$ if and only if $\alpha<\overline{\alpha}$. %, and thus $\mathbb{E}_\gamma[P_{k}]\leq M, \forall k$ if and only if $\alpha<\overline{\alpha}$ according to Lemma~\ref{lm2}. 
			% Since $\lim_{k\rightarrow \infty}\underline{P}_k = \underline{P}$ and $\lim_{k\rightarrow \infty}\overline{P}_k = \overline{P}$, the lower bound $J_\infty^\mathrm{min}=\lim_{N\rightarrow \infty}\frac{1}{N}J_N^\mathrm{min}$ and the upper bound $J_\infty^\mathrm{max}=\lim_{N\rightarrow \infty}\frac{1}{N}J_N^\mathrm{max}$ can be derived as~(\ref{lowerbound}) and~(\ref{upperbound}) based on Theorem~\ref{thm0}, respectively.
			% One can see that when $\gamma>\gamma_c>p_\mathrm{min}$, the solution of equation~(\ref{eq:up}) exist. 
			% According to Lemma~\ref{lm1} and Lemma~\ref{lm2}, MARE~(\ref{eq:s2}) and MARE~(\ref{eq:op}) have a unique positive definite solution if and only if $\eta>\eta_c$ and $\gamma>\gamma_c$, respectively. Based on all of the above, we have Theorem~\ref{thm1}.
		\end{proof}
	}
	\begin{remark}
		%	Note that the left critical value $\underline{\alpha}$ exists only when the system is unstable, i.e., $A$ is unstable. 
		If $\alpha$ is smaller than the left critical value $\underline{\alpha}$, the power for transmitting the control signals will be relatively low. 
		As a result, the cost is unbounded due to inadequate control.
		Similarly, if $\alpha$ is larger than the right critical value $\overline{\alpha}$, the power for transmitting the measurements will be relatively low. Consequently, the cost will be unbounded due to inaccurate estimates. \textcolor{black}{Note that the exact value of $\overline{\alpha}$ cannot be obtained as $\gamma_c$ cannot be calculated exactly as well, but the upper bound and the lower bound of $\overline{\alpha}$ can be found according to the bounds of $\gamma_c$ in Lemma \ref{lm2}.}
	\end{remark}
	
	%\begin{lemma}
	%	$\lim\limits_{k\rightarrow+\infty}g_1^k(\gamma,X)$ is nonincreasing and convex in $\gamma$, and $\lim\limits_{k\rightarrow+\infty}g_2^k(\eta,X)$ is nonincreasing and convex in $\eta$.
	%\end{lemma}
	
	%\begin{definition}
	%	A function $f:\mathbb{R}\rightarrow\mathbb{S}_{++}^n$ is convex in $\mathbb{S}_{++}^n$ if and only if
	%	\begin{equation}
	%		f(\lambda x + (1-\lambda) y)\leq\lambda f(x)+(1-\lambda)f(y).
	%	\end{equation}
	%\end{definition}
	%\begin{lemma}
	%	$S(\alpha)$ is nondecreasing and convex in $\mathbb{S}_{++}^n$ with respect to $\alpha$, and $\overline{P}(\alpha)$ is nonincreasing and convex in $\mathbb{S}_{++}^n$ with respect to $\alpha$.
	%\end{lemma}
	
	\subsection{Optimization Problem}
	\textcolor{black}{When the power splitting ratio $\alpha$ satisfies $\underline{\alpha}<\alpha<\overline{\alpha}$, the LQG cost $J^*_\infty$ can be bounded. However, different values of $\alpha$ result in different costs. 
		% 	Therefore, it is essential to select a value for $\alpha$ such that.
	} 
	In this subsection, an optimization problem is proposed as an aid to the selection of $\alpha$. The intention is to minimize the cost $J^*_\infty$ as much as possible. However, $J^*_\infty$ cannot be minimized directly since it depends on the specific realization of the sequence $\{\gamma_k\}$ and cannot be computed analytically. From the perspective of robustness, therefore, it is reasonable to minimize the upper bound of $J^*_\infty$, i.e., $J^\text{max}_\infty$. 
	%\begin{align*}
	%		\begin{split}
	%			\min_{\alpha,\overline{P},S}J_\infty^\mathrm{max}\triangleq&\mathrm{Tr}\left((A^TSA+W-S)(\overline{P}-\gamma((1-\alpha)p) \overline{P}C^T(C\overline{P}C^T+R)^{-1}C\overline{P})\}\\&+\mathrm{Tr}(SQ),
	%		\end{split}\\
	%		\mathrm{s.t.}\quad S &= A^TSA+W-\eta(\alpha p) A^TSB(B^TSB+U)^{-1}B^TSA,\\
	%		\overline{P} &= A\overline{P}A^T+Q-\gamma((1-\alpha)p) A\overline{P}C^T(C\overline{P}C^T+R)^{-1}C\overline{P}A^T,\\
	%		0&\leq\alpha\leq 1,
	%\end{align*}
	
	Based on the above statement, we have the following optimization problem:
	\begin{align}\label{problem}
		\begin{split}
			\min_{\alpha,\overline{P},S}\quad &J^\text{max}_\infty,
		\end{split}\\
		\begin{split}
			\mathrm{s.t.}\quad &S = A^TSA+W\\&-\eta(h_a\alpha p) A^TSB(B^TSB+U)^{-1}B^TSA,
		\end{split}\label{eq:1}\\
		\begin{split}
			&\overline{P} = A\overline{P}A^T+Q\\&-\gamma(h_e\psi(h_s((1-\alpha)p)) A\overline{P}C^T(C\overline{P}C^T+R)^{-1} \\&\times C\overline{P}A^T,
		\end{split}\label{eq:2}\\
		&0\leq\alpha\leq 1,
	\end{align}
	where $p$ is a positive constant, and $\eta(\cdot): \left[ \right. 0, \infty \left)\right.\rightarrow [0,1]$ and $\gamma(\cdot): \left[ \right. 0, \infty \left)\right.\rightarrow [0,1]$ are monotonically nondecreasing continuous concave functions.
	%\begin{theorem}
	%	$J^\text{max}_\infty$ is convex in $\alpha$.
	%\end{theorem}
	%\begin{proof}
	%	...
	%\end{proof}
	\textcolor{black}{
		\begin{proposition}\label{prop:1}
			Problem (\ref{problem}) is equivalent to the following problem:
			\begin{align}\label{problem2}
				\begin{split}
					\min_{\alpha,\tilde{P},S}\quad &\mathrm{Tr}\left((A^T S A+W-S)\tilde{P}\right)+\mathrm{Tr}\left(SQ\right),
				\end{split}\\
				\mathrm{s.t.}\quad 	&S=\hat{h}\circ\hat{g}_{\alpha}(S),\label{eq:s}\\
				&\tilde{P}=\tilde{g}_{\alpha}\circ\tilde{h}(\tilde{P}),\label{eq:p}\\
				&0\leq\alpha\leq 1.
			\end{align}
			% 		which is feasible and bounded if and only if $\underline{\alpha}\leq\overline{\alpha}$.
		\end{proposition}
		\begin{proof}
                See Appendix~D.
			% It is easy to see that 
			% \begin{align}
			% 	S&=\hat{h}\circ\hat{g}_{\alpha}(S),\\
			% 	\overline{P}&=\tilde{h}\circ\tilde{g}_{\alpha}(\overline{P}).
			% \end{align}
			% Let $\tilde{P}=\tilde{g}_\alpha(\overline{P})$, then we have
			% \begin{equation}
			% 	\begin{split}
			% 		\tilde{P}&=\tilde{g}_\alpha(\tilde{h}\circ\tilde{g}_{\alpha}(\overline{P}))
			% 		\\&=\tilde{g}_\alpha\circ\tilde{h}\circ\tilde{g}_{\alpha}(\overline{P})
			% 		\\&=\tilde{g}_\alpha\circ\tilde{h}(\tilde{P}).
			% 	\end{split}
			% \end{equation}
			% Then it can be seen that problem (\ref{problem}) is equivalent to problem (\ref{problem2}). 
			% % 		If $\overline{\alpha}\leq \underline{\alpha}$, the cost upper bound $J^\text{max}_\infty$ will be unbounded, and thus the optimization problem will be unbounded. Conversely, if the optimization problem has at least one feasible solution, there will be at least one $\alpha$ which lies in the interval $[\underline{\alpha},\overline{\alpha}]$ such that $J^\text{max}_\infty$ is bounded, and hence $\underline{\alpha}\leq\overline{\alpha}$ holds.
			% %		The argument directly follows from Theorem \ref{thm1}.
		\end{proof}
	}
	\begin{remark}
		\textcolor{black}{In this paper, we only consider the situation where problem~(\ref{problem}) is feasible, i.e., $\underline{\alpha}\leq\overline{\alpha}$. Otherwise, the infinite-horizon average cost $J^\text{max}_\infty$ will be unbounded and there is no need to study this problem.}
	\end{remark}
	
	At the first sight, problem (\ref{problem}) is non-convex and it is difficult to convert it to a convex optimization problem. \textcolor{black}{However, its numerical solution can be easily obtained.}
	\textcolor{black}{First, we can discretize $\alpha$ with interval size $\delta$. With $\alpha$ fixed, $S$ and $\tilde{P}$ can be uniquely determined by respectively solving the equations (\ref{eq:s}) and (\ref{eq:p}) via the iterative method, i.e., executing the following iterations until $S$ and $\tilde{P}$ converge to their steady-state values:
		\begin{align}
			S_{k+1}&=\hat{h}\circ\hat{g}_{\alpha}(S_k),\\
			\tilde{P}_{k+1}&=\tilde{g}\alpha\circ\tilde{h}(\tilde{P}_k),
		\end{align}
		where $S_0$ and $\tilde{P}_0$ should be initialized. 
		Note that $S$ and $\tilde{P}$ converge to their steady-state values exponentially fast~\cite{anderson2012optimal}. Therefore, $S$ and $\tilde{P}$ can be calculated by the iterative method efficiently. Then, we can substitute $\alpha$, $S$ and $\tilde{P}$ in the objective function (\ref{problem}) and obtain the corresponding objective value. By iterating through all the discretized $\alpha$, i.e., $\alpha\in\{\underline{\alpha}+k\delta, k = 0,1,\ldots |\underline{\alpha}\leq \underline{\alpha}+k\delta\leq\overline{\alpha}\}$, the value of $\alpha$ which can minimize the value of objective function can be found. 
		The searching algorithm for obtaining the numerical solution of the problem~(\ref{problem}) is summarized in Algorithm \ref{algorithm}.}
	
	\begin{algorithm}
		\caption{Obtain numerical solution of problem (\ref{problem})}
		\begin{algorithmic}\label{algorithm}
			\REQUIRE $\delta, \underline{\alpha}, \overline{\alpha}, P_0, S_0, A, B, C, Q, R, W, U$
			\ENSURE $\alpha$, $J$
			\STATE \textbf{initialize} $\hat{\alpha}=\underline{\alpha}$ $\alpha=\underline{\alpha}$, $J=0$
			\WHILE{$\hat{\alpha}\in[\underline{\alpha},\overline{\alpha}]$}
			\STATE $\tilde{P}=P_0$, $S=S_0$
			\REPEAT
			\STATE $\tilde{P}=\tilde{g}_{\hat{\alpha}}\circ\tilde{h}(\tilde{P})$
			\UNTIL{converge}
			\REPEAT
			\STATE $S=\hat{h}\circ\hat{g}_{\hat{\alpha}}(S)$
			\UNTIL{converge}
			\STATE compute $J^\mathrm{max}_\infty$ according to (\ref{upperbound})
			\IF{$\hat{\alpha}=\underline{\alpha}$}
			\STATE $J=J^\mathrm{max}_\infty$
			\ELSIF{$J^\mathrm{max}_\infty<J$}
			\STATE $J=J^\mathrm{max}_\infty$, $\alpha=\hat{\alpha}$
			%		\ELSE
			%		\STATE
			\ENDIF
			\STATE $\hat{\alpha}=\hat{\alpha}+\delta$
			\ENDWHILE
			\RETURN	
		\end{algorithmic}
	\end{algorithm}
	\begin{remark}
		The accuracy of the obtained numerical solution depends on the selected parameter $\delta$. The optimality gap is bounded by $\delta$, i.e., $|\alpha-\alpha^*|<\delta$, where $\alpha$ is the output of Algorithm \ref{algorithm} and $\alpha^*$ is the optimal solution of problem (\ref{problem}). \textcolor{black}{A smaller $\delta$ brings a higher accuracy, but also leads to more computational overhead.}
	\end{remark}
	
	%\begin{theorem}[Solvability of problem (\ref{flqg})]
	%%	\textbf{(Solvability of problem (\ref{flqg}))}
	%	...
	%\end{theorem}
	
	%\begin{lemma}
	%	Define
	%	\begin{align}
	%		\begin{split}
	%			M_k = &A^TM_{k+1}A+W_k\\&-\eta A^TM_{k+1} B(U_k+B^T(\beta M_{k+1}+\bar{\beta}T_{k+1})B)^{-1}B^TM_{k+1}A,
	%		\end{split}
	%		\\
	%		T_{k+1} = &\gamma A^TM_{k+1}A + \bar{\gamma}A^TT_{k+1}A+W_k,
	%	\end{align}
	%	where $\bar{\gamma} = 1-\gamma$, $\bar{\beta} = \bar{\gamma}\bar{\eta}$, and $\beta = 1 - \bar{\beta}$. Then the optimal control for the UDP-like system is $u_k = -L_k \hat{x}_k$, where 
	%	\begin{equation}
	%		L_k = (U_k + B^T(\beta M_{k+1} + \bar{\beta} T_{k+1})B)^{-1}B^T M_{k+1} A.
	%	\end{equation}
	%\end{lemma}
	%\begin{lemma}
	%	For the TCP-like system, the optimal control is $u_k = -L_k\hat{x}_k$, where 
	%	\begin{equation}
	%		L_k = (U_k + B^TW_{k+1}B)^{-1}B^T W_{k+1} A.
	%	\end{equation}
	%\end{lemma}

	%\subsection{Infinite Horizon LQG Control}

	%\section{Special Case: Static Kalman Gain}
	%Since the sequence $\{\gamma_{k}\}_0^\infty$ is random, the minimal cost $\frac{1}{N}J^*_N$ cannot be computed analytically and does not seem to have a limit. In this section, we investigate the performance of filtering with static Kalman gain, which is easier to be realized in practice due to the reduction of computation. 
	%
	%A natural choice for the estimator gain $K$ is obtained by substituting the error covariance solution of $\bar{P}=\Pi(P,A^T,C^T,Q,R,\gamma)$ into the equation for the Kalman filter gain $K_\gamma=A\bar{P}C^T(C\bar{P}C^T+R)^{-1}$.

	\section{Simulation}
	In this section, we choose $A=1.2$ and $B=C=Q=R=W=U=1$.  
	Moreover, we set channel fading $h_s=h_a=0\text{dB}$ and $h_e=-3\text{dB}$. The transmission power of the transmitter is $p=0.3$mW. 
	In this simulation, we adopt the linear energy harvesting model~\cite{yang2022joint}, i.e.,
	\begin{equation}
		r=\xi (h_s(1-\alpha) p +\sigma_e^2),
	\end{equation}
	where $\sigma_e^2$ denotes the noise power introduced by the receiver antenna and $\xi$ denotes the energy conversion efficiency, which is a constant and characterizes the energy loss of converting the harvested energy to electrical energy. For convenience, $\xi$ is assumed to be $1$.
	In practice, $\sigma_e^2$ is much smaller than $h_s(1-\alpha) p$ so that it can be neglected. Thus, for simplicity, we assume $\sigma_e^2=0$, i.e.,
	\begin{equation}
		r= h_s(1-\alpha) p.
	\end{equation}
	
	The measurements $y_k$ and the control signals $u_k$ are transmitted via the binary phase shifting keying (BPSK) transmission scheme with $B$ bits per packet (symbol). Then we have
	\begin{equation}
		\begin{split}
			\mathbb{P}(\eta_{k}=1|h_a, \alpha, p)=&  \eta(h_a\alpha p)\\=&\left(\int_{-\infty}^{\sqrt{\textcolor{black}{\frac{2h_a\alpha p T_s}{B N_0}}}} \frac{1}{\sqrt{2\pi}}e^{-x^2/2}dx\right)^B,
		\end{split}
	\end{equation}
	and
	\begin{equation}
		\begin{split}
			\mathbb{P}(\gamma_{k}=1|h_e, r)=&  \gamma(h_e r)\\=&\left(\int_{-\infty}^{\sqrt{\textcolor{black}{\frac{2h_e r T_s}{B N_0}}}} \frac{1}{\sqrt{2\pi}}e^{-x^2/2}dx\right)^B,
		\end{split}
	\end{equation}
	where $N_0/2$ is the two-sided noise power spectral density and $T_s$ is the symbol transmission time. We choose $B=2$, $T_s=2\times10^{-7}\text{s}$, and $N_0=2\times10^{-8}\text{W/Hz}$ in the simulations. 
 
	First, we set the time horizon $T=500$ and conduct a Monte Carlo experiment with 1000 runs.
	\begin{figure}[h]
		%	\centering
		%	\includegraphics[width=3in]{stable_UDP}
		%	\caption{Stable UDP-like system }
		%	\label{stable}
		%	%\end{figure}
		%	%\begin{figure}[h]
		\centering
		\includegraphics[width=2.7in]{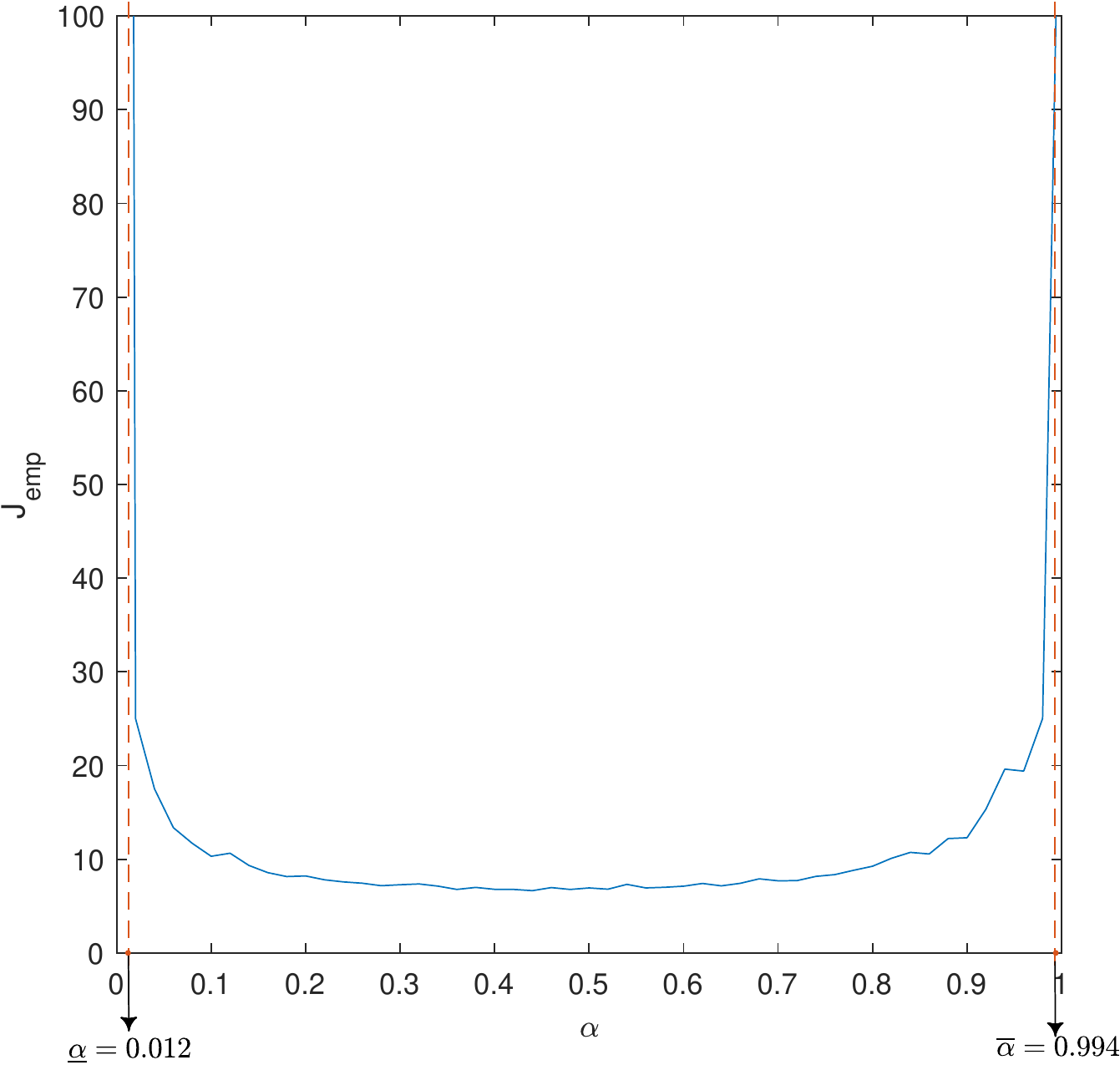}
		\caption{\textcolor{black}{Empirical cost $J_\text{emp}$ for different values of the power splitting ratio $\alpha$}}
		\label{unstable_TCP}
	\end{figure}
	%Fig.\ref{unstable_TCP} shows the relationship between the cost and the ratio $\alpha$. 
	Fig.\ref{unstable_TCP} shows the empirical cost for different values of the power splitting ratio $\alpha$. 
	%For the stable system, there only exists a right critical value for $\alpha$. When $\alpha$ is larger than the right critical value, i.e., the power for recharging the sensor is small enough, the cost is unbounded due to the inaccurate state estimate. 
	%For the unstable system,
	It can be seen that there exist two critical values for the ratio $\alpha$, i.e., $\underline{\alpha}=0.012$ and $\overline{\alpha}=0.994$. If $\alpha$ is less than the left critical value, the power for transmitting the control signal will be relatively low. When $\alpha$ is less than the left critical value or larger than the right critical value, the cost is unbounded due to inadequate control and inaccurate estimate, respectively.
	
	\begin{figure}[h]
		%	\centering
		%	\includegraphics[width=3in]{stable_UDP}
		%	\caption{Stable UDP-like system }
		%	\label{stable}
		%	%\end{figure}
		%	%\begin{figure}[h]
		\centering
		\includegraphics[width=3.4in]{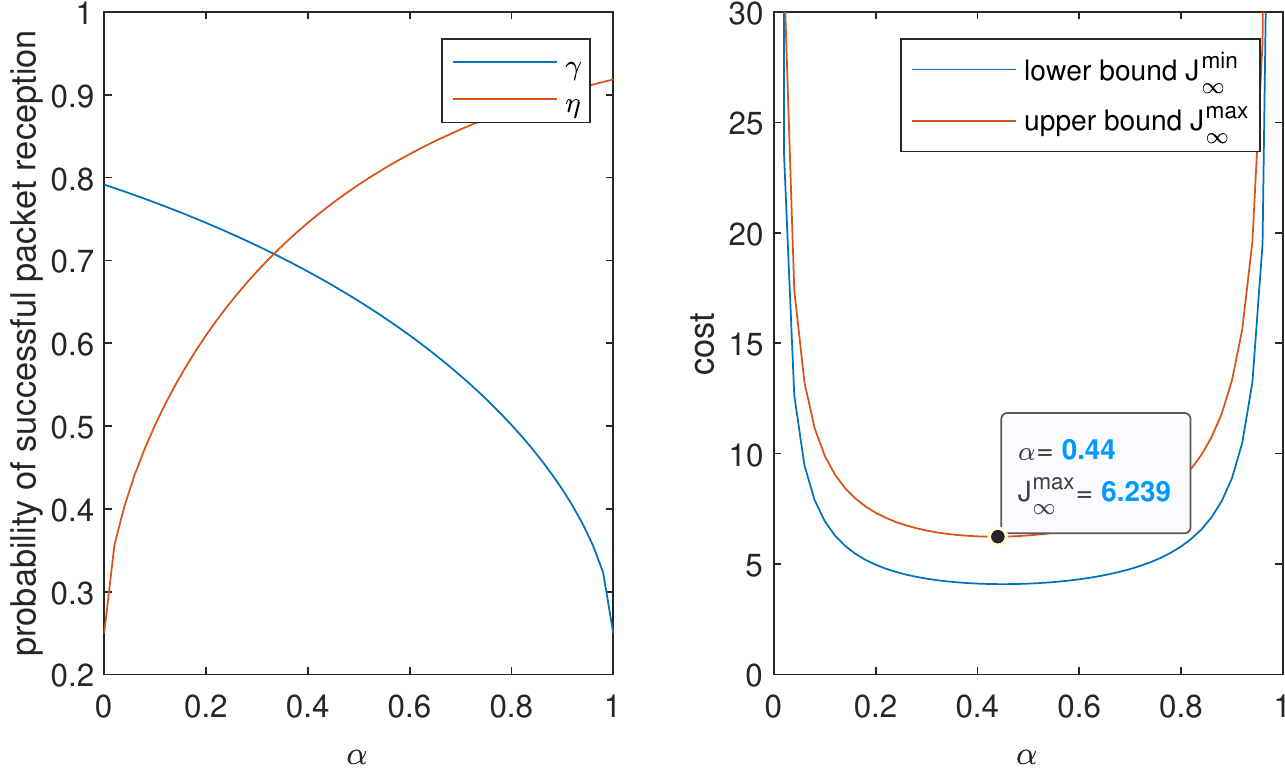}
		\caption{\textcolor{black}{Simulation results with $A=1.2$, $B=Q=R=W=U=1$, $h_s=h_a=0\text{dB}$ and $h_e=-3\text{dB}$.}}
		\label{fig:1}
	\end{figure}
	\begin{figure}[h]
		\centering
		\includegraphics[width=3.4in]{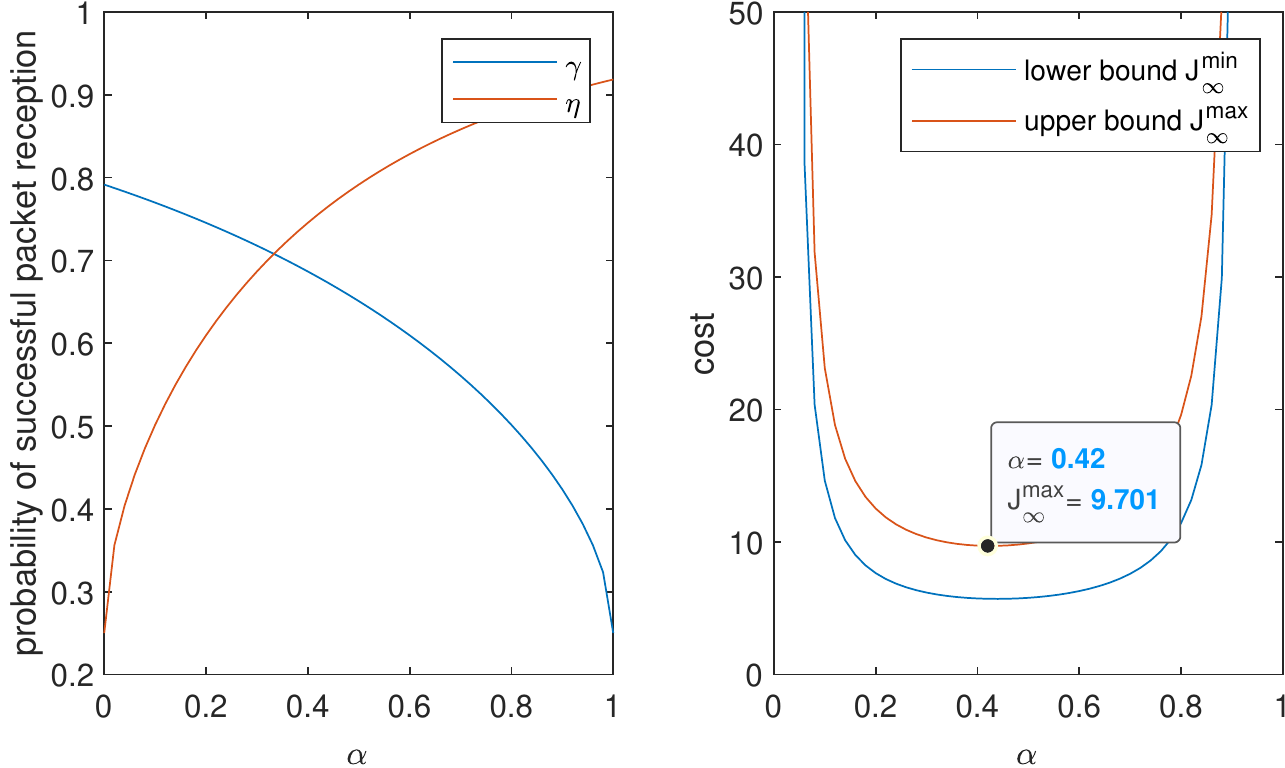}
		\caption{\textcolor{black}{Simulation results with $A=1.3$, $B=Q=R=W=U=1$, $h_s=h_a=0\text{dB}$ and $h_e=-3\text{dB}$.}}
		\label{fig:2}
	\end{figure}
	\begin{figure}[h]
		\centering
		\includegraphics[width=3.4in]{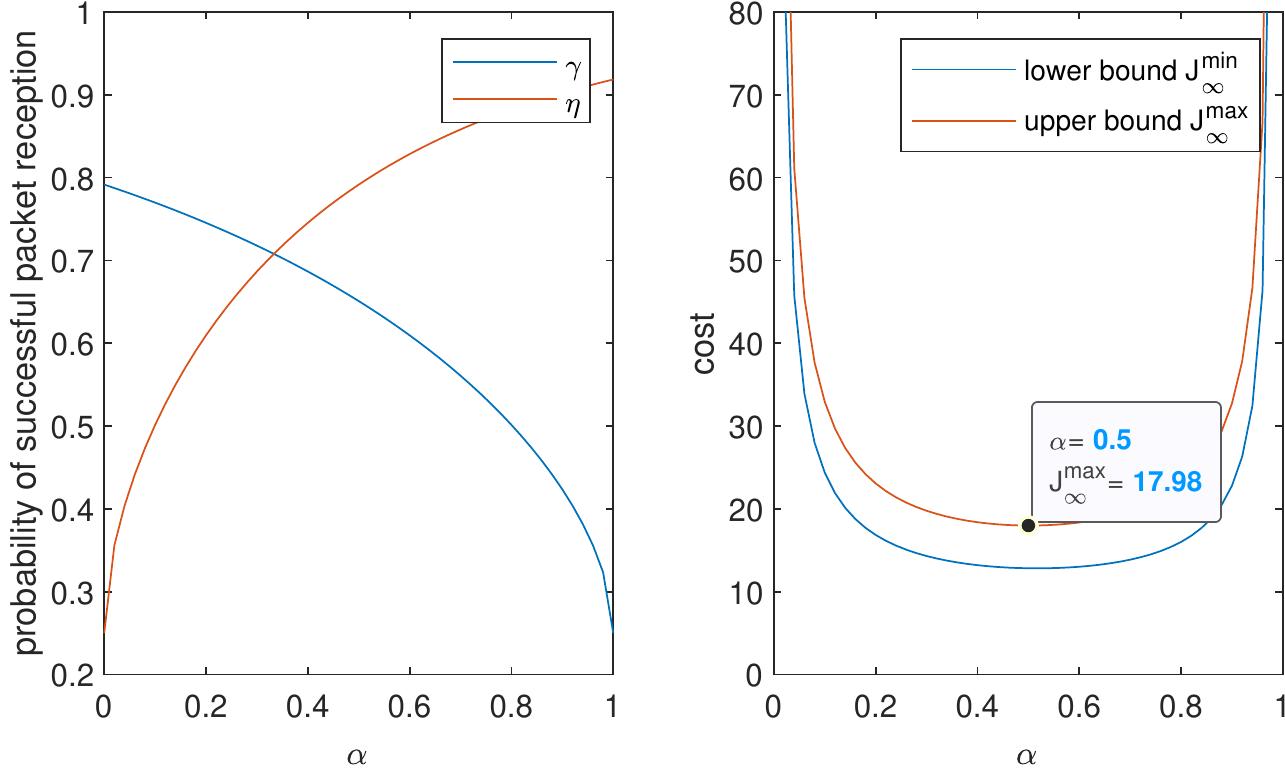}
		\caption{\textcolor{black}{Simulation results with $A=1.2$, $B=Q=R=W=1$, $U=10$, $h_s=h_a=0\text{dB}$ and $h_e=-3\text{dB}$.}}
		\label{fig:3}
	\end{figure}
	\begin{figure}[h]	
		\centering
		\includegraphics[width=3.4in]{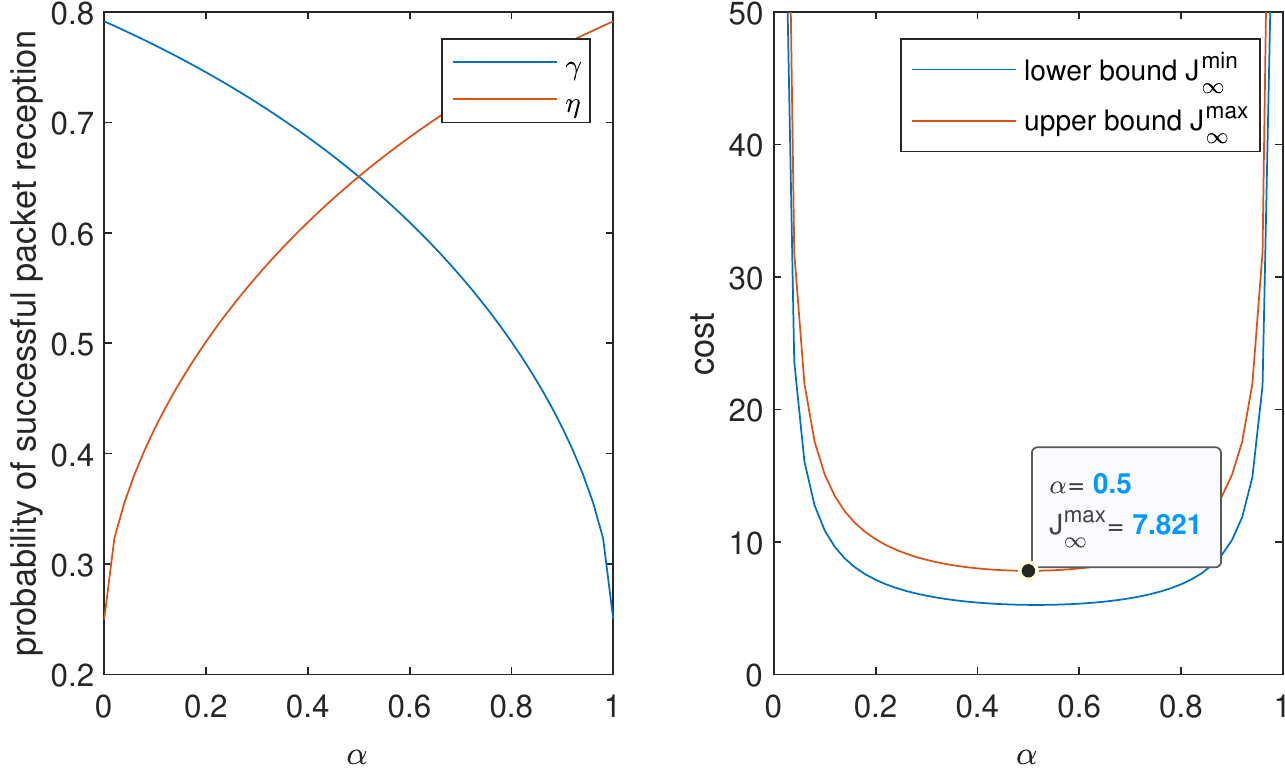}
		\caption{\textcolor{black}{Simulation results with $A=1.2$, $B=Q=R=W=U=1$, $h_s=0\text{dB}$, $h_a=-3\text{dB}$ and $h_e=-3\text{dB}$.}}
		\label{fig:4}
	\end{figure}
	\textcolor{black}{Then, we present some simulation results to illustrate the main results in Section \ref{sec:finite_lqg}. 
		Fig.\ref{fig:1}-\ref{fig:4} show the plot of the probability of successful packet reception (left) and the upper bound $J^\text{max}_\infty$ and the lower bound $J^\text{min}_\infty$ of $J^\text{*}_\infty$ defined in Theorem \ref{thm1} (right) under different simulation parameter settings. Clearly, two critical values for $\alpha$ can also be observed. 
		Fig.\ref{fig:1}-\ref{fig:4} also show the value of $\alpha$ searched by Algorithm \ref{algorithm} ($\delta=0.02$). 
		By comparing Fig.\ref{fig:1} and Fig.\ref{fig:2}, one can see that the stability of the system has an effect on the two critical values of $\alpha$: a more unstable system leads to a lower $\overline{\alpha}$ and a higher $\underline{\alpha}$. 
		By comparing Fig.\ref{fig:1} and Fig.\ref{fig:3}, it can be seen that the value of $\alpha$ minimizing $J^\text{max}_\infty$ will increase when the cost of control is increased. Since inadequate control will cause more cost than inaccurate estimation, the transmitter will tend to use more energy to transmit the control signal to avoid a higher total cost. On the contrary, when the cost of estimation is increased, the transmitter will tend to transfer more energy to the sensor. 
		By comparing Fig.\ref{fig:1} and Fig.\ref{fig:4}, the influence of the channel states can be observed. When the state of the channel from the transmitter to the actuator, i.e., $h_a$, gets worse, the transmitter will use more energy to transmit the control signal to guarantee the control performance. Therefore, the value of $\alpha$ minimizing $J^\text{max}_\infty$ will increase. Similarly, when $h_s$ and $h_e$ get worse, the transmitter will transfer more energy to the sensor to ensure the estimation performance, and the value of $\alpha$ minimizing $J^\text{max}_\infty$ will decrease.
	}
	%\begin{figure}[h]
	%	\centering
	%	\includegraphics[width=3in]{stable_UDP}
	%	\caption{Stable UDP-like system }
	%	\label{stable}
	%%\end{figure}
	%%\begin{figure}[h]
	%	\centering
	%	\includegraphics[width=3in]{unstable_UDP}
	%	\caption{Unstable UDP-like system}
	%	\label{unstable}
	%\end{figure}
	
	\section{Conclusion}
	In this paper, we consider using a novel technology, the so-called SWIPT, to recharge the sensor in the LQG control, which provides a new approach to prolonging the network lifetime. We show that there exist two critical values for the power splitting ratio $\alpha$, and the cost of the infinite horizon LQG control is bounded if and only if $\alpha$ is between these two critical values. Then, we propose an optimization problem to derive the optimal value of $\alpha$. This problem is non-convex but its numerical solution can be derived by our proposed algorithm efficiently. Moreover, we provide the feasible condition of the proposed optimization problem. Simulation results are presented to verify and illustrate the main theoretical results. 
	\textcolor{black}{
		One direction of future work is to consider beamforming when devices are equipped with multiple antennas, which can further improve the spectral efficiency of the system. 
		Another direction is to consider the scenario when there are multiple sensors and multiple actuators.
	}

\begin{appendices}
\section{Proof of Lemma~\ref{lm:1}}
We will prove this lemma by induction. 
The proof consists of two steps:
\begin{enumerate}
    \item[S1.] If $0\leq\alpha_1\leq\alpha_2\leq 1$, then $\hat{h}\circ\hat{g}_{\alpha_1}(X)\leq\hat{h}\circ\hat{g}_{\alpha_2}(X)$;
    \item[S2.] If $(\hat{h}\circ\hat{g}_{\alpha_1})^k(X)\leq(\hat{h}\circ\hat{g}_{\alpha_2})^k(X)$, then $(\hat{h}\circ\hat{g}_{\alpha_1})^{k+1}(X)\leq(\hat{h}\circ\hat{g}_{\alpha_2})^{k+1}(X), \forall 0\leq\alpha_1\leq\alpha_2\leq 1$.
\end{enumerate}
The first step can be easily completed. We will mainly focus on the second step.
\begin{proposition}\label{prop1}
    If $X\geq Y\geq 0$, then $\hat{h}\circ\hat{g}_{\alpha}(X)\geq\hat{h}\circ\hat{g}_{\alpha}(Y)$.
\end{proposition}
If $(\hat{h}\circ\hat{g}_{\alpha_1})^k(X)\leq(\hat{h}\circ\hat{g}_{\alpha_2})^k(X)$, we have
\begin{equation}
    \begin{split}
        (\hat{h}\circ\hat{g}_{\alpha_1})^{k+1}(X)
        &=\hat{h}\circ\hat{g}_{\alpha_1}((\hat{h}\circ\hat{g}_{\alpha_1})^{k}(X))\\&\leq	
        \hat{h}\circ\hat{g}_{\alpha_1}((\hat{h}\circ\hat{g}_{\alpha_2})^{k}(X))\\&\leq\hat{h}\circ\hat{g}_{\alpha_2}((\hat{h}\circ\hat{g}_{\alpha_2})^{k}(X))\\&=(\hat{h}\circ\hat{g}_{\alpha_2})^{k+1}(X),
    \end{split}
\end{equation}
where the first inequality is from proposition \ref{prop1}, and the second inequality is from S1. 

\section{Proof of Theorem~\ref{thm0}}
First, we have the following lemma:
\begin{lemma}%\cite[Lemma 5.4]{schenato2007foundations}
    The expected error covariance matrix $\mathbb{E}_\gamma[P_{k}]$ satisfies the following bounds:
    \begin{equation}
        \tilde{P}_k \leq \mathbb{E}_\gamma[P_{k}] \leq 	\hat{P}_k,
    \end{equation}
where
\begin{align}
    \tilde{P}_k =& (1-\gamma(h_e\psi(h_s((1-\alpha)p))))\underline{P}_k, \\
    \begin{split}
        \hat{P}_k =&  \overline{P}_k\\&-\gamma(h_e\psi(h_s((1-\alpha)p)))\overline{P}_kC^T(C\overline{P}_kC^T+R)^{-1}C\overline{P}_k.
    \end{split}
\end{align}
\end{lemma}

\begin{proof}
    This lemma can be proved based on the observation that the matrices $P_{k+1}^{-}$ and $P_k$ are concave and monotonic functions of $P_k^{-}$. The proof can be easily derived based on Lemma 5.2 in~\cite{schenato2007foundations} and is thus omitted.
\end{proof}
Then, Theorem~\ref{thm0} holds.

\section{Proof of Theorem~\ref{thm1}}
% \begin{proof}
    Following Lemma \ref{lm1} and Lemma \ref{lm2}, $\underline{\alpha}$ and $\overline{\alpha}$ must respectively satisfy equation (\ref{eq1}) and equation (\ref{eq2}). 
    Since $\eta(\cdot)$ and $\gamma(\cdot)$ are monotonically increasing continuous functions, equation (\ref{eq1}) and equation (\ref{eq2}) have one unique solution, respectively. 
    Then we have $\eta>\eta_c$ if and only if $\alpha>\underline{\alpha}$. 
    % According to Lemma~\ref{lm1}, the MARE has a unique strictly positive definite solution $S_\infty$ if and only if $\alpha\leq\underline{\alpha}$. 
    Similarly, we have $\gamma>\gamma_c$ if and only if $\alpha<\overline{\alpha}$. %, and thus $\mathbb{E}_\gamma[P_{k}]\leq M, \forall k$ if and only if $\alpha<\overline{\alpha}$ according to Lemma~\ref{lm2}. 
    Since $\lim_{k\rightarrow \infty}\underline{P}_k = \underline{P}$ and $\lim_{k\rightarrow \infty}\overline{P}_k = \overline{P}$, the lower bound $J_\infty^\mathrm{min}=\lim_{N\rightarrow \infty}\frac{1}{N}J_N^\mathrm{min}$ and the upper bound $J_\infty^\mathrm{max}=\lim_{N\rightarrow \infty}\frac{1}{N}J_N^\mathrm{max}$ can be derived as~(\ref{lowerbound}) and~(\ref{upperbound}) based on Theorem~\ref{thm0}, respectively.
    One can see that when $\gamma>\gamma_c>p_\mathrm{min}$, the solution of equation~(\ref{eq:up}) exist. 
    According to Lemma~\ref{lm1} and Lemma~\ref{lm2}, MARE~(\ref{eq:s2}) and MARE~(\ref{eq:op}) have a unique positive definite solution if and only if $\eta>\eta_c$ and $\gamma>\gamma_c$, respectively. Based on all of the above, we have Theorem~\ref{thm1}.
% \end{proof}

\section{Proof of Proposition~\ref{prop:1}}
It is easy to see that 
    \begin{align}
        S&=\hat{h}\circ\hat{g}_{\alpha}(S),\\
        \overline{P}&=\tilde{h}\circ\tilde{g}_{\alpha}(\overline{P}).
    \end{align}
    Let $\tilde{P}=\tilde{g}_\alpha(\overline{P})$, then we have
    \begin{equation}
        \begin{split}
            \tilde{P}&=\tilde{g}_\alpha(\tilde{h}\circ\tilde{g}_{\alpha}(\overline{P}))
            \\&=\tilde{g}_\alpha\circ\tilde{h}\circ\tilde{g}_{\alpha}(\overline{P})
            \\&=\tilde{g}_\alpha\circ\tilde{h}(\tilde{P}).
        \end{split}
    \end{equation}
    Then it can be seen that problem (\ref{problem}) is equivalent to problem~\eqref{problem2}. 
    
    % 		If $\overline{\alpha}\leq \underline{\alpha}$, the cost upper bound $J^\text{max}_\infty$ will be unbounded, and thus the optimization problem will be unbounded. Conversely, if the optimization problem has at least one feasible solution, there will be at least one $\alpha$ which lies in the interval $[\underline{\alpha},\overline{\alpha}]$ such that $J^\text{max}_\infty$ is bounded, and hence $\underline{\alpha}\leq\overline{\alpha}$ holds.
    %		The argument directly follows from Theorem \ref{thm1}.
\end{appendices}

	\bibliographystyle{ieeetr}
	\bibliography{ref.bib}
	
\end{document}